\documentclass{article}
\usepackage{amssymb,amsmath}
\usepackage{graphicx,enumerate}

\title{The effect of prudence on the optimal allocation in possibilistic and mixed models}

\author{Irina Georgescu \\ \footnotesize Academy of Economic Studies\\ \footnotesize Department of Economic Cybernetics\\ \footnotesize Pia$\c{t}$a Romana No 6  R 70167, Oficiul Postal 22, Bucharest, Romania
 \footnotesize Email: irina.georgescu@csie.ase.ro}

\date{}

\begin{document}
\maketitle

\begin{abstract}
In this paper two portfolio choice models are studied: a purely possibilistic model, in which the return of a risky asset is a fuzzy number, and a mixed model in which a probabilistic background risk is added.
For the two models an approximate formula of the optimal allocation is computed, with respect to the possibilistic moments associated with fuzzy numbers and the indicators of the investor risk preferences (risk aversion, prudence).

\end{abstract}


\newtheorem{definitie}{Definition}[section]
\newtheorem{propozitie}[definitie]{Proposition}
\newtheorem{remarca}[definitie]{Remark}
\newtheorem{exemplu}[definitie]{Example}
\newtheorem{intrebare}[definitie]{Open question}
\newtheorem{lema}[definitie]{Lemma}
\newtheorem{teorema}[definitie]{Theorem}
\newtheorem{corolar}[definitie]{Corollary}

\newenvironment{proof}{\noindent\textbf{Proof.}}{\hfill\rule{2mm}{2mm}\vspace*{5mm}}

\section{Introduction}

The standard portfolio choice problem \cite{arrow}, \cite{pratt} considers the determination of the optimal proportion of the wealth an agent invests in a risk---free asset and in a risky asset. The study of this probabilistic model is usually done in the classical expected utility theory. The optimal allocation of risky asset appears as a solution of a maximization problem. By Taylor approximations several forms of the solution have been found, depending on different moments of the return of the risky asset, as well as some on indicators of the investors's risk preferences. In the form of the solution from \cite{eeckhoudt2}, Chapter $2$ or \cite{gollier}, Chapter $5$, the mean value, the variance and the Arrow-Pratt index of the investor's utility function appear. The approach from \cite{athayde}, \cite{garlappi} led to forms of the approximate solution which depend on the first three moments, the Arrow--Pratt index $r_u$ and the prudence index $P_u$ \cite{kimball}. The solution found in \cite{niguez} is expressed according to the first four moments and the indicators of risk aversion, prudence and temperance of the utility function. Another form of the solution in which the first four moments appear can be found in \cite{courtois}.

In this paper two portfolio choice models are studied: a purely possibilistic model, in which the return of the risky asset is represented by a fuzzy number \cite{carlsson2}, \cite{dubois} and a mixed model, in which a probabilistic background risk appears. In the formulation of the maximization problem for the first model the possibilistic expected utility from \cite{georgescu}, Definition 4.2.7 is used; in case of the mixed model the notion of mixed expected utility from \cite{georgescu}, Definition, 7.1.1 is used. The approximate solutions of the two models will be expressed by the possibilistic moments associated with a fuzzy number (\cite{carlsson2}, \cite{dubois}, \cite{thavaneswaran1}, \cite{thavaneswaran2}) and by the indicators on the investor risk preferences.

In the first part of Section $2$ the definitions of possibilistic expected utility (cf. \cite{georgescu}) and possibilistic indicators of a fuzzy number (expected value, variance,  moments) are presented. The second part of the section contains the definition of a mixed expected utility associated with a mixed vector, a bidimensional utility function and a weighting function (\cite{georgescu}).

Section $3$ is concerned with the possibilistic standard portfolio-choice model, whose construction is inspired by the probabilistic model of \cite{niguez}. The return of the risky asset is here a fuzzy number, while in \cite{niguez} it is a random variable. The total utility function of the model is written as a possibilistic expected value. The maximization problem of the model and the first order conditions are formulated, from which its optimal solution is determined.

Section $4$ is dedicated to the optimal asset allocation in the framework of the possibilistic portfolio model defined in the previous section. Using a second order  Taylor approximation a formula for the approximate calculation of the maximization problem solution is found. In the component of the formula appear  the first three possibilistic moments, the Arrow-Pratt index and the prudence indices of the investor's utility function. The general formula is particularized  for triangular fuzzy numbers and a HARA utility function.

In Section $5$ we consider a mixed portfolio choice model obtained from the possibilistic standard model by adding a possibilistic background  risk. For the mixed model, risk is represented by a bidimensional mixed vector, a component being a fuzzy number and the other a random variable. The agent will have a unidimensional utility function, but the total utility function will be built as a mixed expected utility. We study a maximization problem whose solution is approximated by a calculation formula depending on the first three moments of a fuzzy number, the expected mean value of a random variable and the risk aversion and prudence indicators of a utility function. In the proof of the calculation formula the linearity property of the mixed expected utility is used.

\section{Preliminaries}

In this section we recall some notions and results on the  possibilistic expected utility, mixed expected utility and (cf. \cite{georgescu}) and some possibilistic indicators associated with fuzzy numbers  (cf. \cite{carlsson1}, \cite{carlsson2}, \cite{fuller}, \cite{thavaneswaran1}, \cite{thavaneswaran2}, \cite{zhang}).

\subsection{Possibilistic expected utility}

We fix a mathematical context consisting of:

$\bullet$ a utility function $u$ of class $\mathcal{C}^2$

$\bullet$  a fuzzy number $A$ whose level sets are $[A]^\gamma=[a_1(\gamma), a_2(\gamma)]$, $\gamma \in [0, 1]$.

$\bullet$ a weighting function $f: [0, 1] \rightarrow {\bf{R}}$. ($f$ is a non-negative and increasing function that satisfies $\int_0^1 f(\gamma) d\gamma=1$).

The {\emph{possibilistic expected utility} associated with the triple $(u, A, f)$ is

$E_f(u(A))=\frac{1}{2} \int_0^1 [u(a_1(\gamma))+u(a_2(\gamma))] f(\gamma)d\gamma$            (2.1.1)

The following possibilistic indicators asoociated with a fuzzy number $A$ and a weighting function $f$ are particular cases  of (2.1.1).

$\bullet$ possibilistic expected value \cite{carlsson1}, \cite{fuller}:

 $E_f(A)=\frac{1}{2} \int_0^1 [a_1(\gamma)+a_2(\gamma)]f(\gamma)d\gamma$   (2.1.2)

 ($u$ is the identity function of $\bf{R}$)

$\bullet$ possibilistic variance \cite{carlsson1}, \cite{zhang}:

$Var_f(A)=\frac{1}{2} \int_0^1 [(u(a_1(\gamma))-E_f(A))^2+(u(a_2(\gamma))-E_f(A))^2]f(\gamma)d\gamma$  (2.1.3)

(for $u(x)=(x-E_f(x))^2, x \in {\bf{R}}$).

$\bullet$ the $n$-th order possibilistic moment  \cite{thavaneswaran1}, \cite{thavaneswaran2}:

$M(A^n)=\frac{1}{2} \int_0^1 [u^n(a_1(\gamma))+u^n(a_2(\gamma))]f(\gamma)d\gamma$ (2.1.4)

(for $u(x)=x^n, x \in {\bf{R}}$).

\begin{propozitie}\cite{georgescu}
Let $g: {\bf{R}} \rightarrow {\bf{R}}$, $h: {\bf{R}} \rightarrow {\bf{R}}$ be two utility functions, $a, b \in {\bf{R}}$ and $u=ag+bh$. Then $E_f(u(A))=aE_f(g(A))+bE_f(h(A))$.
\end{propozitie}

\begin{corolar}
$E_f(a+bh(A))=a+bE_f(h(A))$.
\end{corolar}

\subsection{Mixed expected utility}

A bidimensional mixed vector has the form $(A, X)$ where $A$ is a fuzzy number and $X$ is a random variable. We will denote by $M(X)$ the expected value of $X$. If $g: {\bf{R}} \rightarrow {\bf{R}}$ is  a continuous function then $M(g(X))$ is the probabilistic expected utility of $X$ w.r.t. $g$.

Let $u: {\bf{R}}^2 \rightarrow {\bf{R}}$ be a bidimensional utility function of class ${\mathcal{C}}^2$, $(A,X)$ a mixed vector and $f: [0, 1] \rightarrow {\bf{R}}$ a weighting function. Assume that the level sets of the fuzzy number $A$ are $[A]^\gamma=[a_1(\gamma), a_2(\gamma)], \gamma \in [0, 1]$. For any $\gamma \in [0, 1]$, we consider the probabilistic expected values $M(u(a_i(\gamma), X)), i=1,2$.

The mixed expected utility associated with the triple $(u, (A, X), f)$ is:

$E_f(u(A, X))=\frac{1}{2} \int_0^1 [M(u(a_1(\gamma), X))+M(u(a_2(\gamma), X))]f(\gamma)d\gamma$   (2.2.1)

\begin{remarca}
If $a \in {\bf{R}}$ then $E_f(u(a, X))=M(u(a, X))$.
\end{remarca}

\begin{propozitie}
\cite{georgescu} Let $g, h$ be two bidimensional utility functions, $a, b \in {\bf{R}}$  and $u=ag+bh$. Then $E_f(u(A, X))=aE_f(g(A, X))+bE_f(h(A, X))$.
\end{propozitie}

Proposition 2.1 and 2.4 express the linearity of possibilistic expected value and mixed expected utility with respect to the utility functions
which appear in the definitions of these two operators.

\begin{corolar}
If $A$ is a fuzzy number and $Z$ is a random variable then $E_f(AZ)=M(Z)E_f(A)$ and $E_f(A^2Z)=M(Z)E_f(A^2)$.
\end{corolar}

\begin{proof}
We take $u(x, z)=xz$ and applying (2.2.1) we have

$E_f(AZ)=\frac{1}{2} \int_0^1 [M(a_1(\gamma)Z)+M(a_2(\gamma)Z)]f(\gamma)d\gamma=\frac{1}{2} \int_0^1 [a_1(\gamma)M(Z)+a_2(\gamma)M(Z)]f(\gamma)d\gamma=M(Z)E_f(A)$.

Taking $u(x, z)=x^2z$ we obtain

$E_f(A^2Z)=\frac{1}{2} \int_0^1 [M(a_1^2(\gamma)Z)+M(a_2^2(\gamma)Z)]f(\gamma)d\gamma=\frac{1}{2} \int_0^1 [a_1^2(\gamma)M(Z)+a_2^2(\gamma)M(Z)]f(\gamma)d\gamma=M(Z)E_f(A^2)$.
\end{proof}

\section{Possibilistic standard model}

In this section we will present a possibilistic portfolio choice model in which the return of the risky asset is a fuzzy number. In defining the total utility of the model we will use the possibilistic expected utility introduced in the previous section.

We consider an agent (characterized by a utility function $u$ of class $\mathcal{C}^2$, increasing and concave) which invests a wealth $w_0$ in a risk-free asset and in a risky asset. The agent invests the amount $\alpha$ in a risky asset and $w_0-\alpha$ in a risk--free asset. Let $r$ be the return of the risk-free asset and $x$ a value of the return of the risky asset. We denote by $w=w_0(1+r)$ the future wealth of the risk-free strategy, the portfolio value $(w_0-\alpha, \alpha)$ will be (according to \cite{eeckhoudt2}, p. 65-66):

$(w_0-\alpha)(1+r)+\alpha(1+x)=w+\alpha(x-r)$. (3.1)

The probabilistic investment model from \cite{eeckhoudt2}, Chapter $4$ or \cite{gollier}, Chapter $5$ starts from the hypothesis that the return of the risky asset is a random variable $X$. Then $x$ is a value of $X$ and (3.1) leads to the following maximization problem:

$\displaystyle \max_\alpha M[u(w+\alpha(X-r))]$  (3.2)

If we make the assumption that the return of the risky asset is a fuzzy number $B_0$, then $x$ will be a value of $B_0$. To describe the possibilistic model resulting from such a hypothesis, we fix a weighting function $f: [0, 1] \rightarrow {\bf{R}}$. The expression (3.1) suggests to us the following optimization problem:

$\displaystyle \max_\alpha E_f[u(w+\alpha(B_0-r))]$  (3.3)

By denoting with $B=B_0-r$ the excess return, the problem (3.3) becomes:

$\displaystyle \max_\alpha E_f[u(w+\alpha B)]$  (3.4)

Assume that the level sets of the fuzzy number $B$ are $[B]^\gamma=[b_1(\gamma), b_2(\gamma)], \gamma \in [0, 1]$. According to (2.1.1), the total utility function of the model (3.4) will have the following form:

$V(\alpha)=E_f[u(w+\alpha B)]=\frac{1}{2} \int_0^1 [u(w+\alpha b_1(\gamma))+u(w+\alpha b_2(\gamma))]f(\gamma)d\gamma$

Deriving twice one obtains:

$V''(\alpha)=\frac{1}{2} \int_0^1 [b_1^2(\gamma) u''(w+\alpha b_1(\gamma))+b_2^2(\gamma) u''(w+\alpha b_2(\gamma))]f(\gamma)d\gamma$

Since $u'' \leq 0$ it follows $V''(\alpha)\leq 0$ thus $V$ is concave.

We assume everywhere in this paper that the portfolio risk is small, thus analogously with \cite{gollier}, Section 5.2, we can take the excess return $B$ as $B=k\mu +A$ where $\mu >0$ and $A$ is a fuzzy number with $E_f(A)=0$. Of course $E_f(B)=k \mu$ in that case. The total utility $V(\alpha)$ will be written:

$V(\alpha)=E_f[u(w+\alpha(k \mu +A)]$ (3.5)

Assuming that the level sets of $A$ are $[A]^\gamma=[a_1(\gamma), a_2(\gamma)], \gamma \in [0, 1]$, (3.5) becomes:

$V(\alpha)=\frac{1}{2} \int_0^1 [u(w+\alpha(k \mu+a_1(\gamma)))+u(w+\alpha(k \mu+a_2(\gamma)))]f(\gamma)d\gamma$

By deriving one obtains:

$V'(\alpha)=\frac{1}{2} \int_0^1 [(k \mu +a_1(\gamma))u'(w+\alpha(k \mu +a_1(\gamma)))+$

\hspace* {1cm} $(k \mu +a_2(\gamma))u'(w+\alpha(k \mu +a_2(\gamma)))]f(\gamma)d\gamma$

which can be written

$V'(\alpha)=E_f[(k \mu +A)u'(w+\alpha (k\mu +A))]$ (3.6)

Let $\alpha(k)$ be the solution of the maximization problem $\displaystyle \max_\alpha V(\alpha)$, with $V(\alpha)$ being written under the form (3.6). Then the first-order condition $V'(\alpha(k))=0$ will be written:

$E_f[(k\mu +A)u'(w+\alpha(k)(k\mu +A))]=0$  (3.7)

As in \cite{gollier}, Section $5.2$ we will assume that $\alpha(0)=0$.

Everywhere in this paper, we will keep the notations and hypotheses from above.

\section{The effect of prudence on the optimal allocation}

The main result of the section is a formula for the approximate calculation of the solution $\alpha(k)$ of equation (3.7). In the formula will appear
the indicators of absolute risk aversion and prudence marking how these influence the optimal investment level $\alpha(k)$ in the risky asset.

We will consider the second order Taylor approximation of $\alpha(k)$ around $k=0$.

$\alpha(k) \approx \alpha(0)+k\alpha'(0)+ \frac{1}{2} k^2 \alpha''(0)=k \alpha'(0)+ \frac{1}{2} k^2 \alpha''(0)$ (4.1)

For the approximate calculation of $\alpha(k)$ we will determine the approximate values of $\alpha'(k)$ and $\alpha''(k)$\footnote{The calculation of the approximate values of $\alpha'(0)$ and $\alpha''(0)$ follows an analogous line to the one used in \cite{niguez} in the analysis of the probabilistic model. In the proof of the approximate calculation formulas of $\alpha'(0)$ and $\alpha''(0)$ we will use the properties of the possibilistic expected utility from Subsection $2.1$.}. Before this we will recall the Arrow-Pratt index $r_u(w)$ and prudence index $P_u(w)$ associated with the utility function $u$:

$r_u(w)=-\frac{u''(w)}{u'(w)}$; $P_u(w)=-\frac{u'''(w)}{u''(w)}$. (4.2)

\begin{propozitie}
$\alpha'(0) \approx \frac{\mu}{E_f(A^2)}\frac{1}{r_u(w)}$.
\end{propozitie}

\begin{proof}
We consider the Taylor approximation:

$u'(w+\alpha(k\mu+x)) \approx u'(w)+\alpha(k\mu +x)u''(w)$

Then, by (3.6) and Proposition 2.1

$V'(\alpha) \approx E_f[(k \mu+A)(u'(w)+u''(w)\alpha (k\mu +A)]$

\hspace{1cm} $=u'(w)(k\mu +E_f(A))+\alpha u''(w)E_f[(k \mu +A)^2]$.

The equation $V'(\alpha(k))=0$ becomes

$u'(w)(k \mu+E_f(A))+\alpha(k)u''(w)E_f[(k\mu+A)]^2 \approx 0$.

We derive it with respect to $k$:

$u'(w)\mu +u''(w)(\alpha'(k)E_f[(k\mu +A)^2]+2\alpha(k)\mu E_f(k\mu+A))\approx 0$.

In this equality we make $k=0$. Taking into account that $\alpha(0)=0$ it follows

$u'(w)\mu+u''(w) \alpha'(0)E_f(A^2) \approx 0$

from where we determine $\alpha'(0)$:

$\alpha'(0) \approx - \frac{\mu}{E_f(A^2)} \frac{u'(w)}{u''(w)}=\frac{\mu}{E_f(A^2)} \frac{1}{r_u(w)}$

\end{proof}

\begin{propozitie}
$\alpha''(0) \approx \frac{P_u(w)}{(r_u(w))^2} \frac{E_f(A^3)}{(E_f(A^2))^3} \mu^2$.
\end{propozitie}

\begin{proof}
To determine the approximate value of $\alpha''(0)$ we start with the following Taylor approximation:

$u'(w+\alpha(k\mu+x)) \approx u'(w)+\alpha (k\mu +x)u''(w)+\frac{\alpha^2}{2}(k\mu+x)^2u'''(w)$

from which it follows:

$(k \mu+x) u'(w+\alpha(k\mu+x)) \approx u'(w)(k\mu +x) +u''(w) \alpha (k\mu+x)^2+ \frac{u'''(w)}{2}\alpha^2(k\mu+x)$.

Then, by (3.6) and the linearity of the $E_f(.)$ operator:

$V'(\alpha)=E_f[(k \mu+A)u'(w+\alpha (k\mu +A))]$

$\approx u'(w)E_f(k\mu+A)+u''(w)\alpha E_f[(k\mu+A)^2]+\frac{u'''(w)}{2}\alpha^2 E_f[(k\mu+A)^3]$

Using this approximation for $\alpha=\alpha(k)$, the equation $V'(\alpha(k))=0$ becomes

$u'(w)(k\mu+E_f(A))+u''(w)\alpha(k)E_f[(k\mu+A)^2]+\frac{u'''(w)}{2}(\alpha(k))^2E_f[(k\mu+A)^3] \approx 0$.

Deriving with respect to $k$ one obtains:

$\mu u'(w)+u''(w)[\alpha'(k)E_f((k\mu+A)^2)+2 \mu \alpha(k)E_f(k \mu+A)]+$

$+\frac{u'''(w)}{2} [2\alpha(k)\alpha'(k)E_f((k\mu+A)^3)+3(\alpha(k))^2 \mu E_f((k\mu+A)^2)] \approx 0$.

We derive one more time with respect to $k$:

$u''(w)[\alpha''(k)E_f((k\mu+A)^2)+2 \mu \alpha'(k)E_f(k \mu+A)+2 \mu \alpha'(k)E_f(k \mu+A)+$

$ 2 \mu^2 \alpha(k)]+ \frac{u'''(w)}{2}[2(\alpha'(k))^2 E_f((k\mu+A)^3)+2\alpha(k) \alpha''(k)
E_f((k\mu+A)^3)+$

$+6\alpha(k) \alpha'(k) E_f((k\mu+A)^2) + 6 \mu \alpha(k) \alpha'(k) E_f((k \mu +A)^2)+6 \mu^2 (\alpha(k)^2) E_f(k \mu +A)] \approx 0$

In the previous equation we take $k=0$.

$u''(w)[\alpha''(0)E_f(A^2)+ 2\mu \alpha'(0) E_f(A)+ 2 \mu \alpha'(0)E_f(A)+2 \mu^2 \alpha(0)]+$

$+\frac{u'''(w)}{2}[2 (\alpha'(0))^2 E_f(A^3)+2 \alpha(0) \alpha''(0)E_f(A^3) + 6 \alpha(0) \alpha'(0) E_f(A^2)+$

$6 \mu \alpha(0)E_f(A^2) + 6 \mu^2 (\alpha(0))^2 E_f(A)] \approx 0$.

Taking into account that $\alpha(0)=0$ and $E_f(A)=0$ one obtains

$u''(w)\alpha''(0)E_f(A^2) + u'''(w)(\alpha'(0))^2 E_f(A^3) \approx 0$

from where we get $\alpha"(0)$:

$\alpha''(0) \approx - \frac{u'''(w)}{u''(w)} \frac{E_f(A^3)}{E_f(A^2)} (\alpha'(0))^2 $.

By replacing $\alpha'(0)$ with the expression from Proposition 4.1 and taking into account (4.2) it follows:

$\alpha''(0)= \frac{P_u(w)}{((r_u(w))^2} \frac{E_f(A^3)}{(E_f(A^2))^3}\mu^2$

\end{proof}

We recall from Section 3 that $A=B-E_f(B)$. The following result gives us an approximate expression of $\alpha(k)$:

\begin{teorema}
$\alpha(k)\approx \frac{1}{r_u(w)} \frac{E_f(B)}{Var_f(B)}+\frac{1}{2} \frac{P_u(w)}{((r_u(w))^2} \frac{E_f[(B-E_f(B))^3]}{(Var_f(B))^3} (E_f(B))^2$.
\end{teorema}

\begin{proof}
By replacing in (4.1) the approximate values of $\alpha'(0)$ and $\alpha''(0)$ given by Propositions  4.1 and 4.2 and taking into account that $E_f(B)=k \mu$ one obtains:

$\alpha(k) \approx k \alpha'(0) +\frac{1}{2} k^2 \alpha''(0)$

$=\frac{k\mu}{E_f(A^2)} \frac{1}{r_u(w)} +\frac{1}{2} (k \mu)^2 \frac{P_u(w)}{(r_u(w))^2} \frac{E_f(A^3)}{(E_f(A^2))^3}$

$=\frac{E_f(B)}{E_f(A^2)} \frac{1}{r_u(w)} + \frac{1}{2} (E_f(B))^2 \frac{P_u(w)}{(r_u(w))^2}\frac{E_f(A^3)}{(E_f(A^2))^3}$.

But $E_f(A^2)=E_f[(B-E_f(B))^2]=Var_f(B)$.  Then

$\alpha(k) \approx \frac{1}{r_u(w)} \frac{E_f(B)}{Var_f(B)} +\frac{1}{2} \frac{P_u(w)}{(r_u(w))^2} \frac{E_f[(B-E_f(B))^3]}{(Var_f(B))^3} (E_f(B))^2$.
\end{proof}

\begin{remarca}
The previous proposition gives us an approximate solution of the maximization problem $\displaystyle \max_{\alpha} V(\alpha)$ with respect to the indices of absolute risk aversion and prudence $r_u(w)$, $P_u(w)$, and the first three possibilistic moments $E_f(B)$, $Var_f(B)$ and $E_f[(B-E_f(B))^3]$.

This result can be seen as a possibilistic version of the formula (A.6) of \cite{niguez} which gives us the optimal allocation of investment in the context of a probabilistic portfolio choice model.
\end{remarca}

\begin{exemplu}\normalfont
We consider the triangular fuzzy number $B$ defined by:

$B(t)=\left\{\begin{array}{rcl}
1-\frac{b-x}{\alpha}& \mbox{if}& b -\alpha \leq x \leq  b \\
1- \frac{x-b}{\beta} & \mbox{if}& b \leq x \leq b+ \beta\\
0 & \mbox{otherwise}&
\end{array}\right.$

The level sets of $B$ are $[B]^\gamma=[b_1(\gamma), b_2(\gamma)]$, where $b_1(\gamma)=b-(1-\gamma)\alpha$ and $b_2(\gamma)=b+(1-\gamma) \beta$, for $\gamma \in [0, 1]$. We assume that the weighting function $f$ has the form $f(\gamma)=2\gamma$, for $\gamma \in [0, 1]$. Then, by \cite{thavaneswaran2}, Lemma 2.1:

$E_f(B)=b+\frac{\beta-\alpha}{6}$; $Var_f(B)=\frac{\alpha^2+\beta^2+\alpha \beta}{18}$;

$E_f[(B-E_f(B))^2]=\int_0^1 \gamma [ (b_1(\gamma)-E_f(B))^2 +(b_2(\gamma)-E_f(B))^2] d\gamma$

\hspace{2.6cm} $=\frac{19(\beta^3-\alpha^3)}{1080}+\frac{\alpha \beta (\beta-\alpha)}{72}$.

By replacing these indicators in the formula of Theorem 4.3, we obtain

$\alpha(k) \approx \frac{1}{r_u(w)} \frac{b+\frac{\beta-\alpha}{6}}{\frac{\alpha^2+\beta^2+\alpha \beta }{18}} + \frac{1}{2} \frac{P_u(w)}{((r_u(w))^2}
\frac{\frac{19(\beta^3-\alpha^3)}{1080}+ \frac{\alpha \beta (\beta-\alpha)}{72}}{(\frac{\alpha^2+\beta^2+ \alpha \beta}{18})^3} (b+\frac{\beta-\alpha}{6})^2.$

Assume that the utility function $u$ is HARA-type (see \cite{gollier}, Section 3.6):

$u(w)=\zeta (\eta+w)^{1-\gamma}$ for $\eta +\frac{w}{\gamma}>0$.

$r_u(w)=(\eta+\frac{w}{\gamma})^{-1}$; $P_u(w)=\frac{\gamma+1}{\gamma}(\eta+\frac{w}{\gamma})^{-1}$

Then, according to \cite{gollier}, Section 3.6:

$\frac{1}{r_u(w)}=\eta+\frac{w}{\gamma}$ and $\frac{P_u(w)}{((r_u(w))^2}=\frac{\frac{\gamma+1}{\gamma}(\eta+\frac{w}{\gamma})^{-1}}{(\eta+\frac{w}{\gamma})^{-2}}=\frac{\gamma+1}{\gamma}(\eta+ \frac{w}{\gamma})$.

Replacing in the approximation calculation formula of $\alpha(k)$, it follows:

$\alpha(k) \approx (\eta+\frac{w}{\gamma})\frac{b+\frac{\beta-\alpha}{6}}{\frac{\alpha^2+\beta^2+\alpha \beta }{18}}+ \frac{1}{2} \frac{\gamma+1}{\gamma}(\eta+ \frac{w}{\gamma}) \frac{\frac{19(\beta^3-\alpha^3)}{1080}+ \frac{\alpha \beta (\beta-\alpha)}{72}}{(\frac{\alpha^2+\beta^2+ \alpha \beta}{18})^3} (b+\frac{\beta-\alpha}{6})^2$

\end{exemplu}

\section{Possibilistic portfolio choice model with probabilistic background risk}

The portfolio choice model from the previous sections has been built on the hypothesis that the return of the risky asset is modeled by a fuzzy number.
In this section we will study a mixed model in which, besides this possibilistic risk, a probabilistic background risk may appear, modeled by a random variable $Z$. This mixed model comes from the possibilistic standard model by adding $Z$ in the composition of the total utility function. More precisely, the total utility function $W(\alpha)$ will have the following form:

$W(\alpha)=E_f[u(w+ \alpha(k\mu+A)+Z)]$         (5.1)

where the other components of the model have the same meaning as in Section $3$.

Assume that the level sets of $A$ are $[A]^\alpha=[a_1(\gamma), a_2(\gamma)]$, $\gamma \in [0, 1]$. By the definition (2.2.1) of the mixed expected utility,
formula (5.1) can be written:

$W(\alpha)=\frac{1}{2} \int_0^1 [M(u(w+\alpha(k\mu+a_1(\gamma))+Z))+M(u(w+\alpha(k\mu+a_2(\gamma))+Z))]f(\gamma)d\gamma$.

We derive $W(\alpha)$:

$W'(\alpha)=\frac{1}{2} \int_0^1 (k \mu +a_1(\gamma))M(u'(w+\alpha(k\mu+a_1(\gamma))+Z))f(\gamma)d\gamma+$

\hspace{1.5cm} $+ \frac{1}{2} \int_0^1 k \mu +a_2(\gamma))M(u'(w+\alpha(k\mu+a_2(\gamma))+Z))f(\gamma)d\gamma$

$W'(\alpha)$ can be written as:

$W'(\alpha)=E_f[(k\mu+A)u'(w+\alpha(k \mu+A)+Z)]$.  (5.2)

By deriving one more time we obtain:

$W''(\alpha)=E_f[( k \mu+A)^2u''(w+\alpha(k \mu+A)+Z)]$

Since $u''\leq 0$ it follows $W''(\alpha) \leq 0$ thus $W$ is concave. Then the solution $\beta(k)$ of the optimization problem $\displaystyle \max_{\alpha} W(\alpha)$ will be given by $W'(\beta(k))=0$. By (5.2)

$E_f[(k \mu+A)u'(w+\beta(k)(k \mu+A)+Z)]=0$ (5.3)

In this case we will also make the natural hypothesis $\beta(0)=0$.

To compute an approximate value of $\beta(k)$ we will write the second order Taylor approximation of $\beta(k)$ around $k=0$.

$\beta(k) \approx \beta(0) + k \beta'(0)+ \frac{1}{2} k^2 \beta''(0) =k \beta'(0) +\frac{1}{2} k^2 \beta''(0)$  (5.4)

We propose to find some approximate values of $\beta'(0)$ and $\beta''(0)$.

\begin{propozitie}
$\beta'(0) \approx \frac{\mu}{E_f(A^2)} (\frac{1}{r_u(w)}-M(Z))$
\end{propozitie}

\begin{proof}
We consider the Taylor approximation:

$u'(w+\alpha(k \mu+x)+z) \approx u'(w)+(\alpha(k \mu +x)+z)u''(w)$

Then

$(k \mu+x) u'(w+\alpha(k \mu +x)+ z)\approx u'(w)(k \mu +x)+u''(w)\alpha(k \mu +x)^2 +u''(w)z(k \mu +x)$

From this relation, from (5.2) and the linearity of mixed expected utility it follows:

$W'(\alpha) \approx u'(w)(k \mu +E_f(A))+u''(w) \alpha E_f[(k \mu +A)^2]+u''(w) E_f[(k \mu+A)Z]$.

Then the equation $W'(\beta(k))=0$ will be written

$u'(w)(k \mu +E_f(A))+u''(w)\beta(k) E_f[(k \mu+A)^2] +u''(w) E_f[(k \mu +A)Z] \approx 0$.

By deriving with respect to $k$ one obtains:

$u'(w)\mu +u''(w)(\beta'(k)E_f[(k\mu +A)^2]+ 2\beta(k) \mu E_f(k \mu +A)) +u''(w)\mu M(Z) \approx 0$.

For $k=0$ it follows

$u'(w)\mu+u''(w)\mu M(Z)+u''(w)\beta'(0)E_f(A^2) \approx0$.

from where $\beta'(0)$ is obtained:

$\beta'(0) \approx - \frac{(u'(w)+u''(w)M(Z))\mu}{u''(w)E_f(A^2)}=\frac{\mu}{E_f(A^2)}(\frac{1}{r_u(w)}-M(Z))$
\end{proof}

\begin{propozitie}
$\beta''(0)\approx \frac{P_u(w)(\beta'(0))^2}{Var_f(B)} \frac{E_f[(B-E_f(B))^3]}{1-M(Z)P_u(w)}$

\end{propozitie}

\begin{proof}
We consider the Taylor approximation

$u'(w+\alpha(k \mu+x)+z) \approx u'(w)+u''(w) [\alpha(k \mu+x)+z] +\frac{1}{2} u'''(w) [\alpha(k \mu +x)+z]^2$

from where it follows

$(k \mu+x)u'(w+\alpha(k \mu + x)+z) \approx u'(w)(k \mu +x)+u''(w)(k \mu +x)[\alpha(k \mu+x)+z]$

$+\frac{1}{2}u'''(w)(k \mu+x) [\alpha (k \mu+x)+z]^2$.

By (5.2), the previous relation and the linearity of mixed expected utility, we will have

$W'(\alpha) \approx u'(w)(k \mu +E_f(A))+u''(w)E_f[(k \mu +A)(\alpha(k \mu+A)+Z)]+$

$+\frac{1}{2} u'''(w) E_f[(k \mu+A)(\alpha (k \mu+A)+Z)^2]$

Then from $W'(\beta(k))=0$ we will deduce:

$u'(w)(k \mu+E_f(A))+u''(w)E_f[(k \mu +A)(\beta(k)(k \mu+A)+Z)]+$

$+\frac{1}{2}u'''(w)E_f[(k \mu+A)(\beta(k)(k \mu+A)+Z)^2] \approx 0$.

If we denote

$g(k)=E_f[(k \mu+A)(\beta(k)(k \mu+A)+Z)]$  (5.5)

$h(k)=E_f[(k \mu+A)(\beta(k)(k \mu+A)+Z)^2]$  (5.6)

then the previous relation can be written

$u'(w)(k \mu+E_f(A))+u''(w)g(k)+ \frac{1}{2} u'''(w)h(k) \approx 0$

Deriving twice with respect to $k$ we obtain:

$u''(w)g''(k)+\frac{1}{2}u'''(w)h''(k) \approx 0$ (5.7)

We set $k=0$ in (5.6):

$u''(w)g''(0) +\frac{1}{2}u'''(w)h''(0) \approx 0$  (5.8)

{\emph{The computation of $g''(0)$}}  We notice that

$g(k)=\beta(k)E_f[(k\mu+A)^2] +E_f[(k \mu+A)Z]$

By denoting $g_1(k)=\beta(k) E_f[(k\mu +A)^2]$ and $g_2(k)=E_f[(k \mu +A)Z]$ we will have $g(k)=g_1(k)+g_2(k)$. One easily sees that
$g_2''(k)=0$, thus $g''(k)=g_1''(k)+g_2''(k)=g_1''(k)$. We derive $g_1(k)$:

$g_1'(k)=\beta'(k)E_f[(k\mu+A)^2]+2\mu \beta(k)E_f(k \mu+A)$

$=\beta'(k) E_f[(k \mu+A)^2]+2 \mu^2 k \beta(k)$

since $E_f(k \mu+A)=k \mu+E_f(A)=k\mu$. We derive one more time

$g_1''(k)=\beta''(k)E_f[(k \mu+A)^2]+2 \mu\beta'(k) E_f(k \mu+A)+2 \mu^2 [\beta(k)+k \beta'(k)]$

Setting $k=0$ in the previous relation and taking into account that $\beta(0)=E_f(A)=0$ it follows

$g''(0)=\beta''(0)E_f(A^2)$ (5.9)

{\emph{The computation of $h''(0)$}}  We write $h(k)$ as

$h(k)=\beta^2(k) E_f[(k \mu+A)^3] +2 \beta(k) E_f[(k \mu+A)^2Z] +E_f[(k \mu+A)Z^2]$

We denote

$h_1(k)=\beta^2(k) E_f[(k \mu+A)^3]$

$h_2(k)=\beta(k) E_f[(k \mu+A)^2 Z]$

$h_3(k)=E_f[(k \mu+A) Z^2]$

Then $h(k)=h_1(k)+h_2(k)+h_3(k)$. One notices that $h_3''(0)=0$, thus

$h''(0)=h_1''(0)+2 h_2''(0)$    (5.10)

We compute first $h_2''(0)$. One can easily notice that

$h_2''(k)=\beta''(k)E_f[(k \mu+A)^2Z]+2 \beta'(k) \frac{d}{dk}E_f[(k\mu+ A)^2Z]+\beta(k)\frac{d^2}{dk^2}E_f[(k\mu+A)^2Z]$.

Taking into account that

$\frac{d}{dk} E_f[(k \mu+A)^2Z]=2 \mu E_f[(k \mu+A)Z]$ 

and $\beta(0)=0$ 

we deduce

$h_2''(0)=\beta''(0) E_f(A^2Z)+4 \mu \beta'(0)E_f(AZ)$   (5.11)

We will compute $h_1''(0)$. We derive twice $h_1(k)$:

$h_1''(k)=\frac{d^2}{dk^2}(\beta^2(k))E_f[(k \mu+A)^3]+2 \frac{d}{dk}(\beta^2(k))\frac{d}{dk}E_f[(k\mu+A)^3]+$

$+\beta^2(k) \frac{d^2}{dk^2} E_f[(k\mu+A)^3]$.

We compute the following derivatives from the last sum:

$\frac{d}{dk}(\beta^2(k))=2 \beta(k) \beta'(k)$;

$\frac{d^2}{dk^2}(\beta^2(k))=2 [\beta''(k) \beta(k)+(\beta'(k))^2]$;

$\frac{d}{dk} E_f[(k \mu+A)^3]=3 \mu E_f[(k\mu +A)^2]$

Then taking into account $\beta(0)=0$:

$h_1''(0)=2[\beta''(0)\beta(0)+(\beta'(0))^2]E_f(A^3)+2 \beta(0)\beta'(0) 3 \mu E_f[(k\mu +A)^2]+\beta^2(0)=\frac{d}{dk} E_f[(k \mu+A)^3]$

$=2 (\beta'(0))^2E_f(A^3)$    (5.12)

By (5.10), (5.11) and (5.12):

$h''(0)=h_1''(0)+2h_2''(0)=$

$=2(\beta'(0))^2E_f(A^3)+2 (\beta''(0)E_f(A^2Z)+4 \mu \beta'(0) E_f(AZ))$  (5.13)

Replacing in (5.7) the values of $g''(0)$ and $h''(0)$ given by (5.9) and (5.13):

$u''(w)\beta''(0)E_f(A^2)+ \frac{1}{2} u'''(w)[2 (\beta'(0))^2  E_f(A^3)+2 (\beta''(0) E_f(A^2Z)+4 \mu \beta'(0) E_f(AZ))] \approx 0$

from where

$\beta''(0)[u''(w)E_f(A^2) +u'''(w) E_f(A^2 Z)] \approx$

$\approx -\beta'(0) u'''(w)[\beta'(0)E_f(A^3)+4 \mu E_f(AZ)]$

The approximate value of $\beta''(0)$ follows:

$\beta''(0) \approx - u'''(w) \beta'(0) \frac{\beta'(0)E_f(A^3)+4 \mu E_f(AZ)}{u''(w)E_f(A^2)+u'''(w)E_f(A^2Z)}$

According to Corollary 2.5, the expression above which approximates $\beta''(0)$ can be written:

$\beta''(0) \approx - u'''(w) \beta'(0) \frac{\beta'(0) E_f(A^3)+4 \mu M(Z) E_f(A)}{u''(w) E_f(A^2)+M(Z) E_f(A^2)}$

$= - \frac{u'''(w)\beta'(0)}{E_f(A^2)} \frac{\beta'(0) E_f(A^3)}{u''(w)+M(Z)u'''(w)}$

since $E_f(A)=0$. If we replace $A$ with $B-E_f(B)$ one obtains

$\beta''(0) \approx -\frac{(\beta'(0))^2 u'''(w)}{Var_f(B)}\frac{E_f[(B-E_f(B))^3]}{u''(w)+M(Z)u'''(w)}$

\hspace{1.25cm} $=\frac{P_u(w)(\beta'(0))^2}{Var_f(B)} \frac{E_f[(B-E_f(B))^3]}{1-M(Z)P_u(w)}$
\end{proof}

\begin{teorema}
$\beta(k) \approx \frac{E_f(B)}{Var_f(B)} [\frac{1}{r_u(w)}-M(Z)]+$

$\frac{1}{2} P_u(w) [\frac{1}{r_u(w)}-M(Z)]^2 \frac{E_f^2(B)E_f[(B-E_f(B))^3]}{Var_f^3(B)[1-M(Z)P_u(w)]}$
\end{teorema}

\begin{proof}
The approximation formula of $\beta'(0)$ from Proposition 5.1 can be written:

$\beta'(0) \approx \frac{\mu}{Var_f(B)}[ \frac{1}{r_u(w)}-M(Z)]$ (5.14)

According to (5.4), (5.13) and Proposition 5.2

$\beta(k) \approx k \beta'(0) +\frac{1}{2} k^2 \beta''(0)$

\hspace{1cm} $=\frac{\mu k}{Var_f(B)}[\frac{1}{r_u(w)}-M(Z)]+$

\hspace{1cm}$+ \frac{1}{2} P_u(w)\frac{(k\mu)^2E_f[(B-E_f(B))^3]}{Var_f^3(B)[1-M(Z)P_u(w)]}[\frac{1}{r_u(w)}-M(Z)]^2$

Since $\mu k=E_f(B)$ it follows

$\beta(k) \approx \frac{E_f(B)}{Var_f(B)}[\frac{1}{r_u(w)}-M(Z)]+$

$\frac{1}{2} P_u(w)[\frac{1}{r_u(w)}-M(Z)]^2\frac{E_f^2(B)E_f[(B-E_f(B))^3]}{Var_f^3(B)[1-M(Z)P_u(w)]}$

\end{proof}

\begin{remarca}
In the approximate expression of $\beta(k)$ from the previous theorem appear the Arrow index and the prudence index of the utility function $u$, the possibilistic indicators $E_f(B)$, $Var_f(B)$ and the possibilistic expected value $M(Z)$.
\end{remarca}

\end{document}